\documentclass[12pt,a4paper]{article}

\usepackage{amsfonts, amsmath, amssymb, amsthm,url,dsfont, gensymb,multicol}

\usepackage{rasmus2}
\usepackage{natbib}

\usepackage{graphicx,xcolor}
\usepackage[english]{babel}
\usepackage{tikz}
\usepackage{array}

\usepackage{authblk}

\newtheorem{thm}{Theorem}[section]

\newtheorem{defi}[thm]{Definition}
\newtheorem{prop}[thm]{Proposition}

\theoremstyle{definition}
\newtheorem{rem}[thm]{Remark}
\newtheorem{example}[thm]{Example}

\newcommand{\E}{\mathbb{E}}

\def\d{{\rm d}}

\newcommand{\M}{\mathbb{M}}

\newcommand{\Sg}{\mathbb{S}}

\newcommand{\ind}{\mathds{1}}

\begin{document}

\title{A $K$-function for inhomogeneous random measures with
    geometric features}

  \author[1]{Anne Marie Svane}

\author[2]{Hans Jacob Teglbj{\ae}rg Stephensen}

\author[1]{Rasmus Waagepetersen}

\affil[1]{Department of Mathematical Sciences, Aalborg University, Denmark}
\affil[2]{Department of Computer Science, University of Copenhagen, Denmark}

\maketitle

\begin{abstract}
	This paper introduces a $K$-function for assessing
        second-order properties of inhomogeneous random measures generated by marked point processes. The
	marks can be geometric objects like fibers or sets of positive
	volume, and the presented $K$-function takes into
          account geometric
	features of the marks, such as tangent directions of
      fibers. The $K$-function requires an estimate of the inhomogeneous
      density function of the random measure. We introduce 
      parametric estimates for the density function based on
      parametric models that represent large scale features of the inhomogeneous
	random measure. The proposed methodology  is applied to simulated fiber patterns
	as well as a three-dimensional data set of steel fibers in concrete.
\end{abstract}


\section{Introduction}

The $K$-function \citep{ripley:76} continues to be a key tool for studying
interactions in spatial point patterns. For a stationary
point process, the $K$-function times the intensity gives for each distance the expected number of
further points within that distance from a typical point. In \cite{baddeley:moeller:waagepetersen:00} the
$K$-function was extended to inhomogeneous point processes with
non-constant intensity functions but satisfying a certain second-order
translational invariance (second-order intensity-reweighted
stationarity). Point processes can be viewed as integer valued random
measures and the $K$-function has been
generalized to more general random measures corresponding to e.g.\ length measures for
random patterns of fibers \citep{chiu:etal:13} or area measures on germ-grain models
with grains of positive area \citep{gallego:ibanez:simo:16}. While the results in
\cite{chiu:etal:13} are restricted to stationary fiber processes,
\cite{gallego:ibanez:simo:16} adopted ideas from \cite{baddeley:moeller:waagepetersen:00} to obtain a
$K$-function for inhomogeneous random sets. \cite{lieshout:18}
introduced a cross $K$-function for inhomogeneous bivariate random
measures generalizing results for the stationary case in \cite{stoyan:ohser:82}.

 The objective of the current paper is related to
\cite{gallego:ibanez:simo:16} in that we define a $K$-function for 
inhomogeneous random measures. 
 The modeling framework for our results essentially
corresponds to marked point process (germ-grain) models with
observable marks (grains).  This allows
us to derive a simple expression for the proposed $K$-function under a null
model of no interaction.
This is in contrast to
\cite{gallego:ibanez:simo:16} who assume that only the union of grains can
be observed. In some cases the assumption of observable grains is
restrictive but it allows for further theoretical insight. Moreover, since 3D
imaging is becoming more common, the assumption of observable grains
becomes less restrictive. For instance, the overlap of objects in 2D images is often due
to a planar projection of a 3D structure. While we define our
$K$-function in a general setup for germ-grain type random measures,
we pay special attention to the case of inhomogeneous fiber
patterns.

When considering patterns of geometric
objects more complex than points it is also relevant to consider local
geometric features of the objects, e.g.\ tangent directions of fibers. Accordingly, our
proposed $K$-function takes into account such geometric
features. In the special case of fiber patterns, our work is related to
\cite{schwandtke:88} who considers weighted second-order measures and
reduced moment measures where the weights depend on geometric features of the fibers such as tangent directions. However, the
theoretical results of \cite{schwandtke:88} are only obtained in a stationary
setting. Similar to \cite{gallego:ibanez:simo:16},
\cite{schwandtke:88} considers the union of fibers where the
applicability of our results relies on being able to distinguish
fibers.

Our $K$-function requires knowledge of the first-order moment density
of the random measure. \cite{gallego:ibanez:simo:16} suggest
kernel estimation of the density function and \cite{lieshout:18} also
uses kernel estimation in the data example considered. In contrast, we
construct new parametric estimators that are less susceptible to overfitting.

A different approach is considered in \cite{hartung:etal:21} who
develop a  mark-weighted
$K$-function \citep{penttinen:stoyan:henttonen:92} for stationary fiber
patterns. The so-called currents metric is used in \cite{hartung:etal:21}
to measure the similarity of fibers both in terms of distance and tangent
directions. This requires that fibers can be observed in their
entirety. This can be a limitation due to edge effects when fibers
extend outside the observation window. In contrast, in the context of
fiber patterns, we only use local information and we can provide edge
corrected versions of our proposed $K$-function.

Summarizing, we introduce a new $K$-function for inhomogeneous random
measures with geometric features generated by germ-grain models. The $K$-function does not require knowledge of the
  germs and only uses local information about the grains in
  neighborhoods used to define the $K$-function.  We also introduce
  novel parametric estimators for inhomogeneous
first-order moment densities needed for the estimation of the $K$-function. We apply our new $K$-function to
simulated data sets and a data example of a 3D fiber pattern.

\section{Moment measures and $K$-function for inhomogeneous random
  measures with geometric features}

Our basic framework for inhomogeneous random measures is a random
countable collection $Z=\{m_1,m_2,\ldots\}$ of random measures $m_i$ taking values in some space $\M$ of finite measures on
$\R^d \times \Sg$ where $\Sg$ is a metric space 
with metric $d_{\Sg}$. In applications, each $m_i$ will often be a measure associated to a geometric object in $\R^d$ with geometric features of the object represented by points in $\Sg$ as in the following example.
\begin{example}\label{ex:segment_def}
Let $\{\Gamma_1, \Gamma_2,\ldots\}$ be an oriented segment
process. Each
$\Gamma_i$ is a random line segment determined by its midpoint $u_i$,
random direction $\Theta_i \in 
\Sg^{d-1}$, and random length
$L_i$. Letting $\Sg = \Sg^{d-1}$ be the space of possible segment directions,  there is an associated collection of random measures $\{m_1,m_2,\ldots\}$ given by
\[ m_{i}(A \times S) = \int_{\Gamma_i} \ind [x \in A, \Theta_i \in S] \lambda(\dd
x) = \int_{-L_i/2}^{L_i/2}  \ind [u_i + r \Theta_i \in A, \Theta_i \in S] \dd r,\]
where $A\subseteq \R^d$, $S\subseteq \Sg$ and $\lambda$ denotes length measure on the segment.
\end{example}
We return to
	Example~\ref{ex:segment_def} in Section~\ref{sec:resultsformarked} where we also give specific constructions of $Z$.
Other examples that fit our framework is when $m_i$ are volume
measures on objects in $\R^d$ as in \cite{gallego:ibanez:simo:16} or
length measures on weighted fibers as suggested by
\cite{schwandtke:88}. The fiber case generalizes Example \ref{ex:segment_def} and will be discussed further in Section
\ref{sec:fiberprocess}. 

We define a random measure $\Phi$ on $\R^d \times \Sg$ by
\begin{equation}\label{eq:Phidef} 
\Phi( A \times S) = \sum_{m \in Z} m(A \times S), \quad A
  \subseteq \R^d, S \subseteq \Sg. 
  \end{equation}
The first moment measure of $\Phi$ is
\begin{equation}\label{eq:mu} 
\mu(A \times S)= \EE \Phi( A \times S) 
\end{equation}
for $A \subseteq \R^d$ and $S \subseteq \Sg$.
If $\mu $ has a density $\rho(\cdot,\cdot)$ with respect to Lebesgue measure times  a reference measure $\nu$ on $\Sg$, then
\begin{align*}
\int_{\R^d \times \Sg} f(z,s) \mu(\d z \times \dd s) &= \int_{\R^d \times \Sg} f(z,s) \rho(z,s) \dd z  \nu(\dd
s). 
\end{align*}
In the sequel we are going to assume that $\rho$ and hence $\mu$ are
well-defined. Some examples where this is satisfied are given in Section~\ref{sec:firstmomentX}.
If $\rho(\cdot,\cdot)$ is constant we say that $\Phi$ is
homogeneous and otherwise that $\Phi$ is inhomogeneous. Note that
inhomogeneity can be due to spatial variation when $\rho$ varies as a
function of its first argument as well as due to non-uniformity of the density
$\rho(z,\cdot)$, $z \in \R^d$, as a function of the `geometric' argument.

The second-order factorial type moment measure for $\Phi$ is defined as
\[ \alpha^{(2)}( A \times S, B \times T) = \EE
  \sum_{m,m' \in Z}^{\neq} m(A \times S)m'(B \times T)\]
  for $A,B \subseteq \R^d$ and $S,T \subseteq \Sg$. Here $\neq$ over
  the sum indicates that the sum is over pairs of distinct measures
  $m$ and $m'$ in $Z$.

For simplicity of exposition, we here and in the following omit details of measurability. For
example, when writing $A \subseteq \R^d$ it is understood that $A$ is
a Borel measurable subset of $\R^d$. Moreover, we routinely use the
standard measure-theoretical arguments that allows us to establish
integral equalities involving general non-negative functions from the
corresponding equalities involving indicator functions (i.e.\ passing
from indicator functions to simple functions to increasing limits of
simple functions). We use $f$ as generic notation for a non-negative measurable
function where the domain of $f$ will be clear from the context.

\subsection{Reweighted moment measures and $K$-function}\label{sec:K}

Inspired by \cite{baddeley:moeller:waagepetersen:00}, we consider
density-reweighted versions of the first- and second-order moment
measures. First note that the reweighted first-order measure 
\[ \mu_{w}(A \times S)= \EE \sum_{m \in Z} \int_{\R^d \times
    \Sg} \frac{\ind [x \in A,s \in S]}{\rho(x,s)}m(\dd x \times \dd
  s) \]
becomes stationary in the sense $\mu_w(A \times S)=\mu_w((A+h) \times
S)$ for all $h\in \R^d$. 
Next, the 
reweighted second-order measure $\alpha^{(2)}_{w}$ is defined
as
\begin{align}\label{eq:alpha2rw}  &\alpha^{(2)}_{w}(A \times S,B \times T)
                 \nonumber \\=&
  \EE  \sum_{m ,m' \in Z}^{\neq}    \int_{(\R^d
  \times \Sg)^2}   \!\!\! \!\!\! \!\!\! \frac{\ind [(x,s) \in A \times S,(y,t) \in B \times
  T]}{\rho(x,s)\rho(y,t)}
 m(\dd x \times \dd s)m'(\dd y \times
  \dd t). 
\end{align}

We say that $\alpha^{(2)}_{w}$ is stationary if
\[ \alpha^{(2)}_{w}(A \times S,B \times T) =
  \alpha^{(2)}_{w}((A +h)\times S,(B +h)\times T)\]
 for $h \in
  \R^d$. In that case a $K$-function can be defined as
follows.
\begin{defi}
For any observation window $W \subseteq \R^d$ of positive volume
$|W|>0$, and $r_1,r_2 >0$ we define
\begin{align}
& K_w(r_1,r_2;W) \label{eq:weightedK} \\ \nonumber
 = &\frac{1}{|W|} \E \!\!\! \sum_{m,m' \in Z}^{\neq} \int_{ (\R^{d}\times \Sg)^2} \!\!\!\!\!\!\!\!\!
  \frac{\ind[x\in W, \|x  -y\|\leq r_1,
  	d_{\Sg} (s,t) \le r_2] }{\rho(x,s)\rho(y,t)}  m(\d x
  \times \dd s) m'(\d y \times \dd  t)   \\
=& \frac{1}{|W|} \int_{W\times \Sg \times \R^d \times  \Sg} \ind[ \|x  -y\|\leq r_1,
d_{\Sg} (s,t) \le r_2] \alpha_w^{(2)}(\d x \times \d s , \d y \times \d t). \nonumber
\end{align}
If $\af^{(2)}_{w}$ is stationary, then $K(\cdot,\cdot;W)$ does not depend on $W$ and
we define 
\begin{equation*}
K_w(r_1,r_2)=K_w(r_1,r_2;W) 
\end{equation*}
for any choice of $W$.
\end{defi}
In Section~\ref{sec:null} below, we discuss when $\alpha_w^{(2)}$ is stationary and show that $K_w$ has a simple  expression under the null model that the measures in $Z$ are  independent. 

Unfortunately, contrary to the
  case of the $K$-function for point processes, the above defined $K_w$ does not seem to have a
  simple interpretation in terms of Palm expectations. However, from
  the definition \eqref{eq:weightedK} it is clear that the
  $K_w$-function is a measure of similarity of pairs of random measures $m$ and
  $m'$ in terms of spatial distance and geometric
  attributes of $m$ and $m'$.

  \section{Results and examples based on marked point process representation}\label{sec:resultsformarked}

To construct specific models, it is convenient to represent
	$Z$ as a marked point process $X= \{ (u,m)
	\}_{u \in Y}$ where $Y$ is a point process on $\R^d$ and the marks $m$
	are the random measures in $Z$. This viewpoint is natural when $Z$ arises from a germ-grain model by letting $Y$ be the point process formed by the germs and the marks be the measures associated with the grains. For instance in Example \ref{ex:segment_def}, $Y$  would be the point process formed by the midpoints of line segments. We will often refer to the points in $Y$ as center points even if they do not need to have such an interpretation. We denote the intensity of $Y$ by $\rho_Y$.
	
	 In this section, we take this marked point process approach and provide details and examples regarding the moment measures $\mu$ and $\alpha^{(2)}$ and the $K_w$-function by exploiting properties of the marked point process $X$.  

  \subsection{First moment measure}\label{sec:firstmomentX}
  
The first-order moment measure of the marked point process $X$ is
\[ \mu_X( A \times F) = \EE  \sum_{(u,m) \in X} \ind [ u \in A, m \in F]  \]
for $A \subseteq \R^d, F \subseteq \M$. By disintegration we can
express $\mu_X$ as
\[ \mu_X( A \times F) = \int_A \EE_u \ind [m \in F] \rho_Y(u) \dd u, \]
where for a fixed $u$ we can view $\EE_u$ as expectation for a marked point $(u,m)$ conditional
on $u \in Y$.
Then by standard arguments,
\begin{align}  \label{eq:campbellX} \EE \sum_{(u,m) \in X} f(u,m) = 
    \int_{\R^d} \EE_uf(u,m)   \rho_Y(u) \dd u. \end{align}
Combining the definitions \eqref{eq:Phidef}  and  \eqref{eq:mu} of $\Phi$ and $\mu$ and \eqref{eq:campbellX},
\begin{align} \nonumber 
\int_{\R^d \times \Sg} f(z,s) \mu(\d z \times \dd s) & = 
\EE \sum_{(u,m) \in X} \int_{\R^d \times \Sg} f(x,s) m(\dd x
\times \dd s)  \\ \label{eq:mu_int}
&= \int_{\R^d} \E_u \int_{\R^d \times \Sg} f(x,s) m(\dd x
\times \dd s) \rho_Y(u) \d u.
\end{align}
We now explore a few examples regarding the nature of the density $\rho$.

\begin{example}
Suppose $\EE_u m(\cdot)$ has a density $\xi_u$ with respect to
Lebesgue measure times $\nu$ where $\nu$  is a reference measure on $\Sg$ (e.g.\ surface measure when $\Sg$ is the unit sphere). Then
\begin{equation*}
  \EE_u m(A \times S) = \int_{\R^d \times \Sg} \ind [(x,s) \in A \times S] \xi_u(x,s) \dd x \nu(\dd
  s) , \quad A \subseteq \R^d, S \subseteq \Sg.
 \end{equation*}
By  \eqref{eq:mu_int},
\begin{align*} \mu( A \times S)  
= & \int_{\R^d \times \R^d  \times \Sg} \ind [ z \in A, s \in S]
                                \xi_u( z,s) \dd z \nu(\dd s)
                                \rho_Y(u) \dd u \\  = &
\int_{\R^d \times \Sg} \ind [ z \in A, s \in S] \int_{\R^d} \xi_u(z,s) 
                                \rho_Y(u) \dd u \dd z \nu(\dd s) 
\end{align*}
meaning that $\mu $ has density
\begin{equation}\label{eq:rho} \rho(z,s) = \int_{\R^d } \xi_u(z,s) \rho_Y(u) \dd u
 \end{equation}
with respect to Lebesgue measure times $\nu$.
  Note that \eqref{eq:rho} may not be finite, in which case the measure $\mu$
and density $\rho$ are not well-defined. A sufficient condition for $\rho$ to be finite is if $\rho_Y$ is
bounded, $\xi_u(z,s)=\xi(z-u,s)$ for some density $\xi$ not depending on
$u$, and 
\[ \int_{\R^d} \xi(z,s) \dd z  < \infty, \quad s \in \Sg. \]
\end{example}
The density $\xi_u$ may not always exist, as the next example shows, but a density for $\mu $ may still exist. 

\begin{example}\label{ex:xidoesnotexist} Consider the segment process in Example \ref{ex:segment_def} and let $Y$ be the point process formed by segment midpoints. For a given $u\in Y $, the corresponding mark $m$ was 
\[ m(A \times S) = \int_{-L/2}^{L/2}  \ind [ u+ r \Theta \in A, \Theta \in S] \dd r.\]
Up to length, both the segment and its direction are uniquely determined by $\Ta$. Due to
this degeneracy, a density for $\EE_u$ with respect to Lebesgue measure on $\R^{d}$ times surface measure $\nu$ on $\Sg^{d-1}$ does
not exist. 

Nevertheless, assume that given $Y$, the pairs $(L_i,\Theta_i)$ are i.i.d. Moreover, assume that $L_i$ and $\Theta_i$ are
independent and $\Theta_i $ has density $\eta$ with respect to $\nu$. In this case, we can still identify a density
for $\mu$: 
\[ \rho(z,s)= \E  \int_{-L/2}^{L/2} \eta(s)  \rho_Y(z-r
  s) \dd r .\]
Details are given in \ref{sec:simplefibermodel}.
\end{example}

\subsection{Second moment measure}\label{sec:secondmomentX}

We define a second-order factorial type moment measure for $X$ by
\[ \alpha_X^{(2)}( A \times F , B \times G) = \EE
  \sum_{(u,m),(u',m') \in X}^{\neq} \ind [ u \in A, m \in F, u' \in B, m'
  \in G] \]
  for $A,B \subseteq \R^d$ and $F,G \subseteq \M$. 
By disintegration and assuming existence of the pair
correlation function of $Y$ \cite[e.g.\ ][]{MoellerWaagepetersenl2004}, we can express $\alpha_X^{(2)}$ as
\begin{equation*} \alpha_X^{(2)}( A \times F , B \times G) = \int_{A \times B}
  \rho_Y(u)\rho_Y(u') g_Y(u,u') \EE_{u,u'}\ind [m \in F,m' \in G] \dd u \dd
  u', \end{equation*}
where $g_Y$ is the pair correlation function and $\EE_{u,u'}$ can be interpreted as mean with respect to the
conditional distribution  of $(m,m')$ for marked
points $(u,m)$ and $(u',m')$ given that $u,u' \in Y$. By standard arguments,
\begin{equation}\label{eq:campbell_alpha}
 \EE \!\!\! \!\!\! 
\sum_{\substack{(u,m),\\(u',m') \in X}}^{\neq} \!\!\! \!\!\!  f(u,m,u',m') 
 = \int_{\R^{2d}} \!\!\! \!\!\! 
\rho_Y(u)\rho_Y(u') g_Y(u,u') \EE_{u,u'}f(u, m,u', m' ) \dd u \dd
u'. \end{equation}

Assume $\EE_{u,u'}m(\cdot)m'(\cdot)$ has density $\xi_{u,u'}$, i.e.\ 
\begin{align} &\EE_{u,u'}m(A\times S)m'(B\times T) \nonumber\\ = &\int_{(\R^d \times \Sg)^2}
\ind [(x,s) \in A \times S, (y,t)
\in B \times T] \xi_{u,u'}((x,s),(y,t)) \dd x
  \dd y \nu(\dd s) \nu(\dd t). \label{eq:xiuu'}\end{align}
with $\nu$ a reference measure on $\Sg$ as in the previous section.
Then 
\begin{align*}
&\alpha^{(2)}( A \times S, B \times T)\\ = &\int_{(\R^d)^2 }
  \rho_Y(u)\rho_Y(u') g_Y(u,u') \int_{(\R^d \times \Sg)^2} \ind [(x,s) \in A \times S, (y,t)
\in B \times T] \\ &\xi_{u,u'}((x,s),(y,t)) \dd x
  \dd y \nu(\dd s) \nu(\dd t) \dd u \dd u' \\
= & \int_{(\R^d \times \Sg)^2} \ind [(z,s) \in A \times S, (z',t)
\in B \times T] \\ 
& \int_{(\R^d)^2}  \rho_Y(u)\rho_Y(u') g_Y(u,u') \xi_{u,u'}((z,s),(z',t)) \dd u \dd u'  \nu(\dd s) \nu(\dd t) \dd z \dd z'.
\end{align*}
It follows that $\alpha^{(2)}$ has density
\begin{equation}\label{eq:rho2} \rho^{(2)}(z,s,z',t) = \int_{\R^d \times \R^d}
\rho_Y(u)\rho_Y(u') g_Y(u,u')  \xi_{u,u'}((z,s),(z',t)) \dd
u \dd u'.   \end{equation}
We return to this expression in Section~\ref{sec:null}.

\subsection{$K$-function under null model}\label{sec:null}

Consider the null model that $Y$ is Poisson and that given $Y$, the
marks are independent (but not necessarily identically distributed). Then the $K$-function has a simple expression:
\begin{prop}\label{prop:null}
	Under the null  model,  $\alpha^{(2)}_w$ is stationary and 
	 \begin{equation}\label{eq:K_null}
	 K_w(r_1,r_2)=K_0(r_1,r_2) := |b(0,r_1)|\nu^2((s,t)\in \mathbb{S}\times \Sg \, |\, d_{\Sg}(s,t)\leq r_2).
	 \end{equation}
\end{prop}
\begin{proof} 
	Under the null model, $g_Y=1$ and $\EE_{u,u'}\ind[m \in F,m'
        \in G]= \EE_{u}\ind[m \in F] \EE_{u'} \ind[m' \in G]$. It follows from \eqref{eq:campbell_alpha} that $\alpha^{(2)}_w = \mu_w \otimes \mu_w$. Hence
	\begin{align*}
	K_w(r_1,r_2){}& = \frac{1}{|W|}\int_{(\R^d\times \mathbb{S})^2} \mathds{1}[z\in W,\|z-z'\|\leq r_1, d_{\mathbb{S}}(s,t)\leq r_2]  \d z \d z' \nu(\d s)\nu(\d t) \\& =|b(0,r_1)|\nu^2((s,t)\in \Sg \times \Sg\, | \, d_{\Sg}(s,t)\leq r_1).
	\end{align*}
\end{proof}
Apart from the null model, it is hard to identify specific models for which $\alpha_w^{(2)}$ is stationary.
Assuming that  $\alpha^{(2)}$ has density $\rho^{(2)}$, we may define
 \[ g(z,s,z',t)= \frac{\rho^{(2)}(z,s,z',t)}{\rho(z,s)\rho(z',t)} \]
as a normalized version of $\rho^{(2)}$ in the spirit of the pair correlation function of a point process. 
 Then stationarity of $\alpha_{w}^{(2)}$ becomes equivalent to
\begin{equation}\label{eq:gtrans} g(z,s,z',t)=g(z+h,s,z'+h,t), \quad
  h \in \R^d.\end{equation}
It is not easy in general to identify specific choices of $X$ for which 
\eqref{eq:gtrans} holds due to the expressions
\eqref{eq:rho} and \eqref{eq:rho2} for $\rho$
and $\rho^{(2)}$ that confound properties of $Y$ and the mark
conditional densities.

\begin{rem} \label{remDiffFibers}
In the definition  \eqref{eq:weightedK} of $K_w$, the sum is only
taken over pairs of distinct measures $m$ and $m'$. Taking the sum over all
pairs instead would result in an extra term in \eqref{eq:K_null} which
in the marked point process setup is
	\begin{align*}
	&\int_{\R^d} \rho_Y(u) \E_u \!\int_{(\R^d\times \mathbb{S})^2}\hspace{-0.8cm}
	\frac{\ind[x\in W, \|x  -y\|\leq r_1,
		d_{\Sg} (s,t) \le r_2]} {\rho(x,s)\rho(y,t)} m(\d x\times \d s) m (\d y \times \d t) \d u.
	\end{align*}
	This depends on the actual mark distribution in a non-trivial way, so we no longer get a simple expression for the $K$-function under the null model. This is a well-known problem even in the stationary case,
	\cite[see e.g.\ Ex.\ 8.4 in][]{chiu:etal:13}. Excluding
	identical pairs was suggested in
	\cite{jensen:kieu:gundersen:90} for the stationary
	case. However, it has been customary to include all pairs
	since this avoids the computational problem of
    distinguishing between different marks.
\end{rem}

\begin{rem}\label{rem:w2}
Consider the marked point process setup. Then
another way of defining a reweighted $K$-function is to replace the weights $\rho(x,s)$ by
	$w(u,x,s)= \xi_u(x,s)\rho_Y(u)$.  This results in a larger
        class of models for which $\alpha^{(2)}_w$ is stationary and the $K$-function hence is well-defined, see
        the discussion in \ref{sec:w2}. The resulting $K$-function
        under the null  model is the same as in \eqref{eq:K_null} up
        to a constant that depends on the support of $\xi_u$.
        However, this choice of weights has  several drawbacks. It requires that the densities
  $\xi_u$ exist and have finite support.  Explicit knowledge of the center points is also required, which may cause problems e.g.\ if some center points lie outside the observation window.
 Moreover, the estimation of $\xi_u$ is more involved, and since
 values of $\xi_u$ far from the center point may be small,
 using them for the $K$-function may be numerically unstable. We
 therefore do not consider these weights further in this paper. 
\end{rem}

\section{Specialization to fibers}\label{sec:fiberprocess}

The case where $X$ is defined by a fiber process is an example of
  special interest for our later application in
  Section~\ref{sec:dataexample}. So consider a fiber
  process, i.e.\ a marked point process $\{(u,\gamma)\}_{u\in
      Y}$  where for a marked point $(u,\gamma)$, $\gamma$ is a finite length smooth
  fiber with midpoint at $u$. In the case of
  oriented fibers, $\Sg$ is the space $\Sg^{d-1}$ of unit vectors in
  $\R^d$, and for unoriented fibers, we let $\Sg$ be a
    hemi-sphere in $\R^d$ such that a point $s\in \Sg$ uniquely identifies a
  line in $\R^d$ passing
  through the origin. 
	To each fiber $\gm $, we associate the measure
	\begin{equation}\label{eq:mgamma} 
	m( A \times S) = \int_{\gm} \ind[ x \in A, \tau(x) \in S] \lambda(
	\dd x) 
	\end{equation}
	where $\tau(x) \in \Sg$ is a unit tangent vector at
        $x$ and $\ld$ denotes the length measure  on the
        fiber. The tangent vector $\tau(x)$ is not uniquely defined for
          unoriented fibers and we simply choose the one ($\tau(x)$ or
          $-\tau(x)$) belonging to
          the hemi-sphere $\Sg$. With a slight abuse of notation we use
$X$ to denote both the fiber process $\{(u,\gamma)\}_{u\in Y}$ and the marked point process $\{(u,m)\}_{u\in Y}$ with
marks given by \eqref{eq:mgamma}. For simplicity of exposition, we consider only the case of oriented fibers below, the modification to the unoriented case being straightforward.

In the oriented case, $\Sg=\Sg^{d-1}$ is a metric space with a metric $d_{\angle}$ given by the angle between unit vectors
\begin{equation*}
d_{\angle}(\tau_1,\tau_2) = \arccos( \tau_1^\T \tau_2 )
\end{equation*}
for $\tau_1,\tau_2\in \Sg^{d-1}$. In the unoriented case, $\Sg$ has the metric $d$ given by
\begin{equation*}
d(\tau_1,\tau_2) = \min \{ d_{\angle}(\tau_1,\tau_2), d_{\angle}(\tau_1,-\tau_2)\},
\end{equation*}
which can be interpreted as the angle between the lines determined by $\tau_1$ and $\tau_2$.
The measure $\nu$ on $\Sg$ is the 
surface area measure normalized so that $\nu(\Sg)=1$.

The $K$-function from Section~\ref{sec:K}
becomes
\begin{align*} 
& K_{w}(r_1,r_2;W)  \\ = & \frac{1}{|W|} \E
                        \sum_{\substack{(u,\gm),\\(u',\gm') \in
  X}}^{\neq}  \int_{\gm}\int_{\gm'} \frac{ \ind[x\in W, \|x  -y\|\leq r_1,
  d_{\Sg} (\tau(x),\tau(y)) \le r_2] } {\rho(x,\tau(x))\rho(y,\tau(y))} \lambda(\dd y) \lambda(\dd
  x).\nonumber                  
\end{align*}
The value of the $K$-function under the null  model is 
\begin{equation*}
K_{0}(r_1,r_2)=  |b(0,r_1)| \nu(\tau \in \Sg \mid  d_{\Sg}(\tau,e_1)\leq r_2),
\end{equation*}
where $e_1\in \Sg$ is the north pole.
In the oriented case, for $d=2$,
\begin{equation*}
K_0(r_1,r_2) = r_1^2 r_2, \quad r_1\geq 0,\, r_2\in [0,\pi],
\end{equation*}
and for
$d=3$,
\begin{align*}
K_0(r_1,r_2){}&= \frac{2\pi}{3} r_1^3 (1-\cos r_2)\quad r_1\geq 0,\,  r_2\in [0,\pi].
\end{align*}
In the unoriented case, $K_0$ is twice of $K_0$ in the oriented case for $r_1\geq 0, r_2\in [0,\pi/2]$.

\section{Estimation}

In this section we discuss estimation of $K_w$ given an
observation of $Z$ restricted to a bounded observation window $W$. We
first consider the case of a known
density $\rho$ and next we consider estimation of $\rho$ based on
$\Phi$ restricted to $W$.

\subsection{Estimation of $K_w$}
Let $W$ be  a bounded observation window of positive volume.
For $m \in Z$, suppose we only observe the restriction of  $m$ to $W\times \Sg$. Assume for now that the density $\rho$ is known. We then estimate the $K$-function by
\begin{align}
& \hat K_{w}(r_1,r_2;W) \nonumber \\ = & \frac{1}{|W|} 
                        \sum_{m,m' \in
  Z}^{\neq}  \int_{(\R^d \times \Sg)^2}
\frac{\ind[x\in W,y\in W]}{\rho(x,s)\rho(y,t)}
e(x,y)  m(\dd x \times \dd s) m(\dd y \times \dd t)  \label{eq:Kest}
\end{align}
where $e(x,y)$ is an edge correction factor. Compared
to the definition \eqref{eq:weightedK} we only include locations $y$ that are
observed within $W$. Therefore an edge correction $e$ is introduced to ensure
unbiasedness of the estimate. Different types of edge corrections are
possible but we use here
\[ e(x,y ) = \frac{|W|}{|W\cap W_{y-x}|}\]
where $W_{y-x}$ is $W$ translated by the vector
$y-x$. These edge corrections make \eqref{eq:Kest} unbiased whenever $\alpha_w^{(2)}$ is stationary and they are easily evaluated when $W$
is a hyper rectangle. 

\subsubsection{The fiber case}\label{sec:Kfibercase}

Assume that we observe a fiber process $X$ inside an observation
window $W$, i.e. we observe $\cup_{(u,\gm) \in X}  \gm \cap W$. The
estimator \eqref{eq:Kest} then becomes
\begin{align}
& \hat K_{w}(r_1,r_2;W) \nonumber \\\nonumber
 = & \frac{1}{|W|} 
                        \sum_{\substack{(u,\gm),\\(u',\gm') \in
  X}}^{\neq}  \int_{\gm}\int_{\gm'} \frac{\ind[x\in W,y\in W, \|x  -y\|\leq r_1,
d_{\Sg} (\tau(x),\tau(y)) \le r_2] } {\rho(x,\tau(x))\rho(y,\tau(y))} \\
& \cdot e(x,y)  \lambda(\dd y) \lambda(\dd
  x). \label{eq:Kestimatefiber}
\end{align}
In practice line integrals are replaced by Monte Carlo
estimates, 
\begin{equation} \label{eq:montecarlo} \int_{\gm} f(x,\tau(x)) \ld( \dd x) \approx \frac{1}{\phi}\sum_{x \in X_\gm}
  f(x,\tau(x)) \end{equation}
where $X_{\gm}$ is a point process on $\gm$ of intensity $\phi$.

In applications, $\rho$ is typically not known and must be estimated. We consider
estimation of $\rho$ in the next section.

\subsection{Estimating $\rho$}\label{sec:rhoestimation}

Given a single realization of an inhomogeneous random measure $\Phi$ it is not possible to uniquely decompose variation in the
pattern into `systematic' or `large scale' variation represented by
the density $\rho$ and random `local' variation quantified by the
$K_w$-function. The use of kernel smoothing for estimation of the
density $\rho$ is therefore problematic because an arbitrary
proportion of the total
variation can be absorbed into a kernel estimate of $\rho$ depending
on the level of smoothing applied. This is a well-known problem for
estimation of $K$-functions for inhomogeneous point processes, see
e.g.\ \cite{baddeley:moeller:waagepetersen:00} or
\cite{shaw:moeller:waagepetersen:20}. In these papers it is shown that
the use of a kernel intensity estimate leads to bias in the subsequent
estimation of the $K$-function. Due to these considerations we propose to use a
simple parametric model for the density $\rho$ which is confined to
capture large scale variation in $\Phi$.

\subsubsection{Parametric model for  $\rho$}\label{sec:rhomodel}
We consider a model for
$\rho(z,s)$ of the form
\begin{equation}\label{eq:rhomodel} \rho(z,s;\beta) = \rho(z;\bt) \eta(s)\end{equation}
where $\eta(\cdot)$ is a probability density on $\Sg$ with respect to the measure $\nu$ and $\rho(\cdot;\bt)$ is a linear
or log-linear regression model. 

Suppose for instance
$\rho_Y(u)=\beta_0 + \bt^\T u$ is a
linear trend model and assume that $\xi_u$ exists and is reflection invariant around $u$
  in the sense $\xi_u(z,s)=\xi(z-u,s)$ for a density $\xi$ with
  $\xi(-z,s)=\xi(z,s)$. Then, as shown in \ref{sec:reflectioninvariant}, $\rho(z,s)$ has the
form  \eqref{eq:rhomodel} with $\rho(z;\bt)=l \rho_Y(z)$ for some
positive constant $l$.
 More generally, if 
$\rho_Y$ is a smooth  function and $\xi_u(\cdot,s)$, $u
  \in \R^d$, is
of bounded support and reflection invariant, then $\rho(z,s)$ is well
approximated by the right hand side of \eqref{eq:rhomodel}, see \ref{sec:reflectioninvariant} for details. 
For the segment process models used for simulations in
Section~\ref{sec:simulations}, an expression for $\rho(z,s)$ of the
form \eqref{eq:rhomodel} is also obtained, see Example~\ref{ex:xidoesnotexist}, \ref{sec:simplefibermodel} and \ref{sec:simplefibermodelinteract}. 

The density $\eta$ can be estimated directly from the observed
marks as demonstrated in Section~\ref{sec:dataexample}. In the next section we propose estimating equations for the
parameter $\beta$.

\subsubsection{Estimating equation for $\beta$}\label{sec:rho_estimation}

Assume $\Phi$ is observed on $W$.
Inspired by composite likelihood methods for point
processes \cite[see e.g.\ the review in][]{moeller:waagepetersen:17}, we may consider an estimating function of the form
\begin{align*} e(\beta) & =  \int_{\R^d \times \Sg}  \ind[x\in
                          W] f(x,s;\bt) \Phi(\dd x \times \dd s)  \\
  & -  \int_{\R^d \times \Sg}  \ind[z\in W] f(z,s;\bt) \rho(z,s;\bt) \dd z   \nu(\dd s) \end{align*}
for some appropriate choice of $f$. Note that this type of
  estimating function is applicable for general $\Phi$ not necessarily
  of the form \eqref{eq:Phidef}. In terms of $Z$, 
\begin{align*}  \int_{\R^d \times \Sg}  \ind[x\in W] f(x,s;\bt) \Phi(\dd x \times \dd s) 
= \sum_{m \in Z} \int_{\R^d \times \Sg}  \ind[x\in W] f(x,s;\bt) m(\dd x \times \dd s).
\end{align*}

In case of the linear intensity 
\[ \rho(z,s;\beta) = \eta(s)(1,z^\T)  \beta, \] 
we might use $f(z,s;\bt)= \eta(s)^{-1}\dd
\rho(z,s;\beta)/\dd \beta = (1,z^\T)^\T$ so that 
\begin{align*} 
e(\bt) =  \int_{\R^d \times \Sg}  \ind[x\in W]
  (1,x^\T)^\T \Phi(\dd x \times \dd s) -  \int_{\R^d}  \ind[z\in W]
  (1,z^\T)^\T (1,z^\T)   \dd z  \beta         .
\end{align*}
Letting
\[ R = \int_{\R^d}  \ind[z\in W] (1,z^\T)^\T (1,z^\T) \dd z \text{ and } L = \int_{\R^d \times \Sg} \ind[x\in W] (1,x^\T)^\T \Phi(\dd
    x \times \dd s) \]
we obtain
\[ \hat \beta = R^{-1} L \]
provided $R$ is invertible.

For instance, in the case $d=3$ and $W=[0,a] \times [0,b] \times [0,c]$,
$R$ has the simple expression
\[ R=  |W| \begin{bmatrix} 1 & a/2 & b/2 & c/2\\ a/2 & a^2/3 &
    ab/4 & ac/4 \\ b/2  & ab/4  & b^2 / 3 & bc/4\\
    c/2 & ac/4 &  bc/4 & c^2/3\end{bmatrix}. \]

\subsubsection{The fiber case}\label{sec:fibercaserho}

A fiber $\gm$ will in practice  be represented by points $X_\gm$ sampled
randomly on $\gm$ with intensity $\phi$. Thereby we obtain the unbiased estimate
\[ \hat L = \frac{1}{\phi} \sum_{(u,\gm) \in X} \sum_{x \in X_\gm} \ind[x\in W] (1,x^\T)^\T \]
of $L$. The resulting estimating function $\hat L - R \bt$ can then be
viewed as a first-order estimating function based on the point process
$\cup_{(u,\gm) \in X} X_\gm$ of intensity $\phi \rho
(\cdot;\beta)$. We then obtain the unbiased estimate $\hat \beta=
R^{-1} \hat L$. 

If $\rho(\cdot;\bt)$ is a log-linear model,  we can use
composite likelihood methods \cite[for example][]{waagepetersen:07,moeller:waagepetersen:07} for the point
process $\cup_{(u,\gm) \in X}X_\gm$ to estimate $\phi
\rho(\cdot;\bt)$. In practice, this can be done using the \texttt{spatstat}
\citep{baddeley:rubak:turner:15} procedure \texttt{ppm}.

\section{Simulation studies}\label{sec:simulations}

We have studied the performance of the $K$-function using
  simulations of fiber patterns in two settings: the null model and a
  model with dependency between fiber directions.

\subsection{The null model}\label{sec:null_sim}

We first considered 10,000 simulations under a null model (see
  Section~\ref{sec:null} and \ref{sec:simplefibermodel}) for fiber patterns with
  independent fibers. We simulated the patterns on a window of size
  $20 \times 20$ where the fibers  were generated as
  i.i.d.\ line segments having midpoint at the origin with random uniform orientation and random uniform length between 0 and 2. The
  fibers were then  translated to have midpoints on an inhomogeneous Poisson process with intensity
  decreasing linearly from 3.5 to 0.5 from left to right. An example
  of such a fiber pattern is given in the top left plot of Figure~\ref{fig:Knull}.
\begin{figure}
	\includegraphics[width=\textwidth]{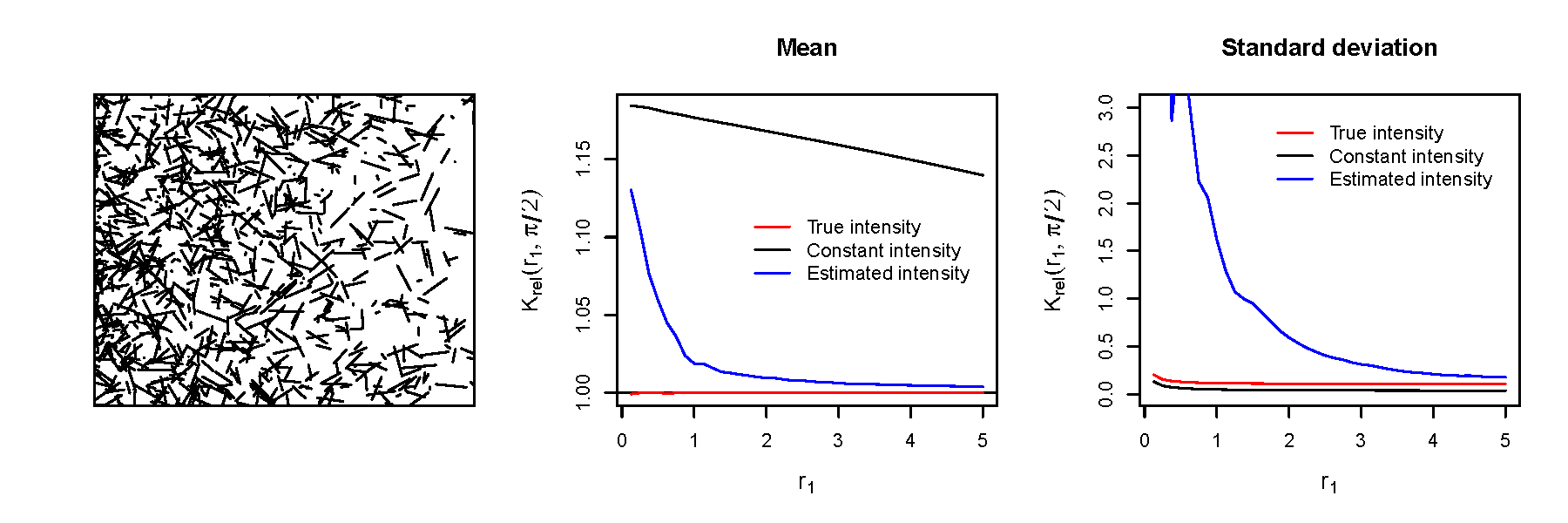}\\
	\includegraphics[width=\textwidth]{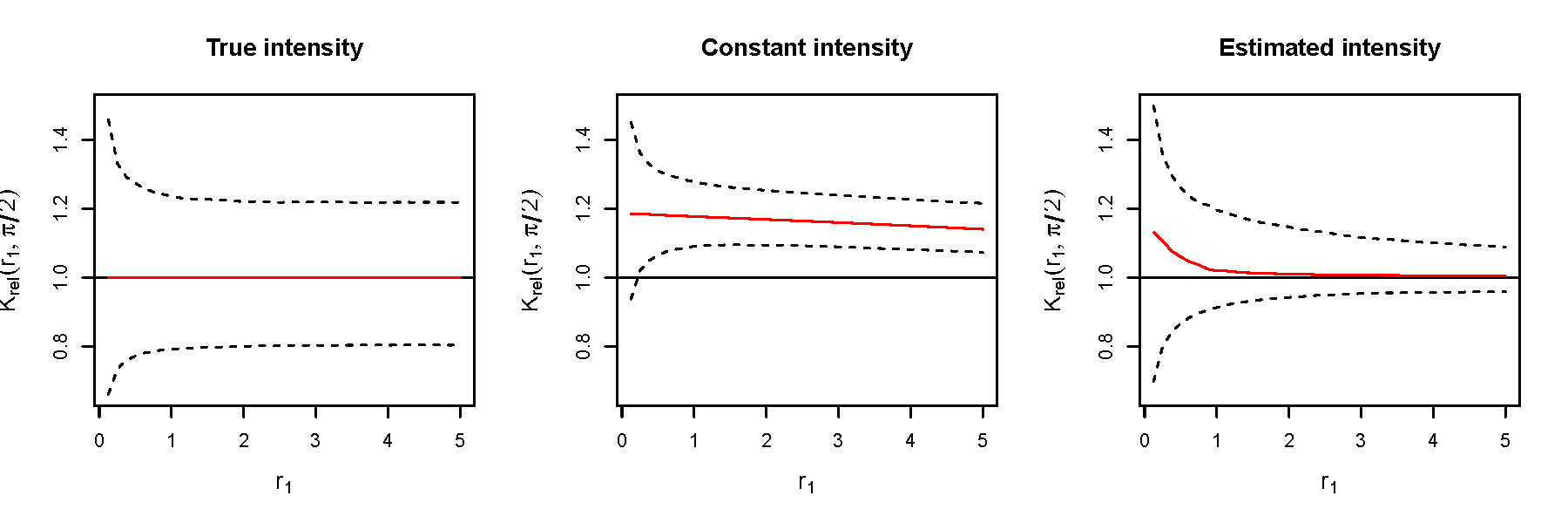}
	\caption{Top left: an example of a simulated fiber
          pattern. Top middle/right: mean/standard deviation of 10,000
          estimates of the relative $K$-function \eqref{eq:relK} when
          the intensity is either known, assumed constant, or
          estimated linear. Lower row: mean (red curve) and 95\%
          probability bands (dashed curves) for estimates of the relative $K$-function
          for the three choices of intensity functions.}\label{fig:Knull}
\end{figure}

For each simulation, we estimated the $K_w$-function by
  \eqref{eq:Kestimatefiber} using for $\rho$ either 1) the true
  theoretical linear
  expression as given in
  \ref{sec:simplefibermodel},  2) an estimate assuming a
  constant intensity, or 3) the linear expression with $\beta$
  estimated as described in Section~\ref{sec:rho_estimation} and with
 $\eta$ the true uniform density. The pointwise means, standard
 deviations, and 95\% probability bands for the relative $K$-function
\begin{equation}\label{eq:relK}
K_{\text{rel}}(r_1,r_2)=\frac{K_w(r_1,r_2)}{K_0(r_1,r_2)}
\end{equation}   
 with $r_2=\tfrac{\pi}{2}$ are shown in Figure \ref{fig:Knull}. We see
 that using the true intensity, the estimator becomes unbiased, as
 expected by theory. Estimating the intensity linearly results in a
 small positive bias, while using a constant intensity yields a larger
 bias of around 0.15. This bias arises because large scale variation in
 $\rho$ is picked up as random interaction by the $K_w$-function when a
 constant intensity is assumed.

 The probability bands show that when the intensity is estimated, the
 resulting curves are typically closer to the mean than when the intensity is
 known. This phenomenon is also known from $K$-functions for point
 processes, see e.g.\ \cite{waagepetersen:guan:09}. Although most
 curves are close to the mean when the intensity
 is estimated linearly, a few outliers occur which is reflected in the standard
 deviations. This is most pronounced for small values of $r_1$ where estimated
 values of $K_{rel}(r_1,\tfrac{\pi}{2})$ range from $-200$ to
 $1000$. This happens because for a few simulations, the estimated
 intensity comes close to zero, or even below zero, resulting in very
 small estimates of $\rho(z,s)$ for some $z$.

\subsection{Dependent fibers}

Next, we simulated 1000 realizations from a model with dependent
fibers on the same window as before and using the same inhomogeneous
Poisson process $Y$ for the fiber center points. To obtain dependent
fibers we simulated two independent Gaussian random fields with exponential correlation function and
evaluated these at the points of $Y$. For each $u$ in $Y$, the
associated values of the Gaussian fields give the $x$- and
$y$-coordinates of one endpoint of a centered line segment while the
negation of this point forms the other endpoint resulting in a
fiber symmetric around the origin. The centered fiber is
  next translated to have midpoint $u$. The random fields are
scaled to obtain a mean fiber length of 1 as in the simulations with
independent fibers. This model is described in more detail in \ref{sec:simplefibermodelinteract} where it
is shown that this model has $\rho$ of the form \eqref{eq:rhomodel}
with $\eta$ a uniform density and $\rho(\cdot;\beta)$ a linear model.

An example of a resulting simulated fiber pattern is shown to the left
in Figure~\ref{fig:aligned}. Visually, it is not easy to determine
exactly what is the nature of the difference between the simulations
from the null model and from the model with dependent fibers. This is,
however, reflected in the $K_w$-function. For the estimation, we
replaced the unknown density function with an estimate as in option 3)
 for the null model in Section \ref{sec:null_sim}. The mean of the estimated relative $K$-function is shown in the right plot of
Figure~\ref{fig:aligned}. The mean estimated
$K_{rel}(r_1,\tfrac{\pi}{2})$ does not deviate much from the value
under the null model. However, the smaller the value of $r_2$, the larger the $K_w$-function becomes compared to the null 
model. This suggests that while the amount of fibers around a point on a given fiber is the same as under independence, fibers tend to be more aligned than under the null model. This effect is present for all $r_1$, but most pronounced for small values, suggesting that the correlation mainly is present on shorter scales.
\begin{figure}
	\includegraphics[width=\textwidth]{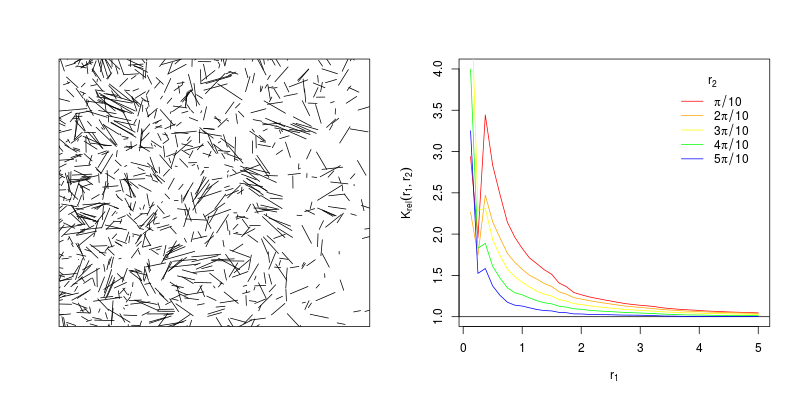}\\
	\caption{Left: a simulation of an inhomogeneous fiber
          pattern with correlated fibers. Right: the mean of the estimated
          relative $K$-function \eqref{eq:relK} as a function of $r_1$ for
            selected values of $r_2$.}\label{fig:aligned}
\end{figure}

\section{Data example}\label{sec:dataexample}

In this section we apply our $K$-function to a 3D fiber pattern, shown in Figure \ref{fig:concrete}, of
steel fibers in ultra high performance fiber-reinforced concrete. The
fiber pattern was extracted from a series of micro computed tomography
images kindly
provided by the authors of \cite{maryamh:etal:21}. The objective in
\cite{maryamh:etal:21} is to study the relationship between production
parameters of the fiber-reinforced concrete, the spatial distribution
and orientation of the steel fibers, and mechanical properties of the concrete (results of bending tests
applied to concrete specimens). More specifically, the data set
considered here was obtained by imaging a $160 \times 40\times 40$mm
block of concrete with a fiber volume percentage of 1\% and with
fibers of length 12.5mm and of diameter 0.3mm. Each image voxel was of
sidelength 0.0906mm. Ideally, for maximal bending strength, the
fibers should be aligned with the long axis of the block. Due to
sedimentation, spatial inhomogeneity of the fibers is possible
although the degree of inhomogeneity depends on the viscosity of the
concrete used for casting the concrete block, see
\cite{maryamh:etal:21} who considered several blocks produced with
varying production parameters.  \cite{maryamh:etal:21} studied the
spatial variation in fiber density using area fraction profiles
obtained from the microscope images. For the directions of fibers they
obtained partial derivatives of the gray scale images and considered
eigenvectors of the resulting orientation matrices. For these
approaches it is not necessary to identify individual fibers. To take
the analysis a step further we identified the individual fibers in
order to apply our new $K_w$-function. 
\begin{figure}
	\centering
	\includegraphics[width=\textwidth]{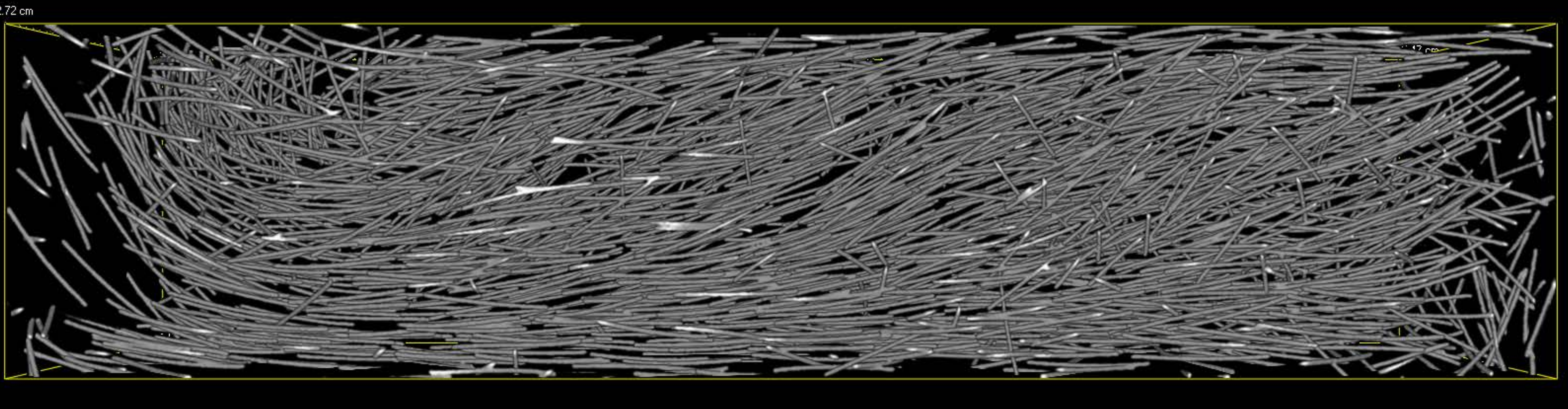}
	\caption{3D image of a subblock of the fiber-reinforced concrete dataset.}\label{fig:concrete}
\end{figure}

\subsection{Fiber segmentation}

Since pairs of fibers are often quite close, an initial segmentation
mask failed to separate all fibers. To correct this, a manual
separation was performed using the ImageJ
\citep{schneider:rasband:eliceiri:12} software by removing pixels that
connected fibers. Each fiber was then represented by a cloud of pixel
center points and the fiber cloud points were assigned fiber specific labels
using a connected component algorithm implemented in the Python
package SciPy \citep{2020SciPy-NMeth}. Some clouds contained very
small numbers of fiber points and were discarded as artefacts of noise
(in total 8\% of the fiber points were discarded). To represent each
fiber as a curve in space, we next obtained a least squares fit of a
vector function $p(t) = (p_x(t), p_y(t), p_z(t))^\T$ to each fiber
point cloud. Here the functions $p_x,p_y,p_z$ were chosen to be third
order polynomials. As mentioned in Section~\ref{sec:Kfibercase}, in
practice we approximate integration along fibers with Monte Carlo
integration. Specifically, we randomly generated a point process of
equispaced points on each fiber. The spacing was 0.906mm corresponding
to 10 image voxel side lenghts. 
Figure~\ref{fig:fiber1} shows
an example of a fiber point cloud, the fitted curve, and the Monte
Carlo points on the fiber curve. Tangent directions were computed at each Monte Carlo point by normalizing the derivative of $p(t)$. Since fibers
	are unoriented we, following Section~\ref{sec:fiberprocess}, represented tangents by points on
	be the hemi-sphere  $\Sg$ with north pole at $(0,1,0)$.
Due to strong edge effects in each
end of the block, we follow \cite{maryamh:etal:21} and discard data
from the first and last (referring to the long axis) 2 cm of the
block. 
\begin{figure}
\centering
 \includegraphics[width=0.6\textwidth]{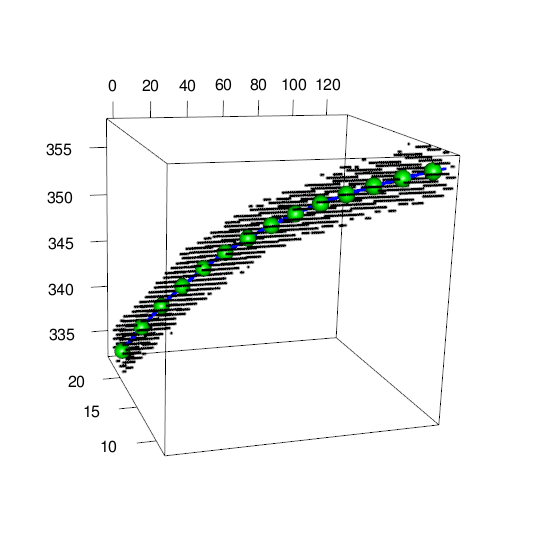}
  \caption{Example of a fiber point set, the fitted fiber curve
    (blue) and Monte Carlo sample points (green).}\label{fig:fiber1}
\end{figure}

\subsection{Estimation of the density function}\label{sec:density_est}

The density function $\rho$ was fitted following the approach in
Section~\ref{sec:rhoestimation} using the model \eqref{eq:rhomodel}
with a linear model for the space
component $\rho(u;\beta)=(1,u^\T)^\T \bt$, $u \in \R^3$, $\bt \in \R^4$, as suggested by \ref{sec:reflectioninvariant}. Following
Section~\ref{sec:fibercaserho}, we obtained the estimate 
$\hat{\beta}=(1.82, 0.0027,-0.0061, -0.045)^\T$.
 Relative to the intercept,
the slopes are of moderate magnitude meaning that the volume density of the
fibers is only varying moderately over space (e.g.\ an estimated
increase in intensity of 0.32 (120 times 0.0027) across the block in the $x$-direction). 

  For the probability density $\eta$ for tangent directions on
    the hemi-sphere $\Sg$, we converted each unit tangent vector
$\tau=(\tau_1,\tau_2,\tau_3)$ into
cylindrical coordinates $(h,\phi)=(\tau_1,\arctan(\tau_3/\tau_2)) \in
[-1,1] \times [-\pi/2,\pi/2]$
(with $\arctan(0/0)=0$ and $\arctan(\pm \infty) = \pm \pi/2$.)
An exploratory analysis of the
height and angle cylindrical coordinates suggested to treat these as
outcomes of independent random variables and assume a height distribution symmetric around zero. We hence modelled the joint
density of heights and angles as a product of densities that were
non-parametrically estimated by histograms, see Figure~\ref{fig:fiber2}. Finally, the Jacobian of the mapping from
the unit hemi-sphere to cylindrical coordinates is one, so the density
$\eta$ is simply estimated by the estimated joint density of height
and angle. 
\begin{figure}
\centering
\begin{tabular}{cc}
\includegraphics[width=0.4\textwidth]{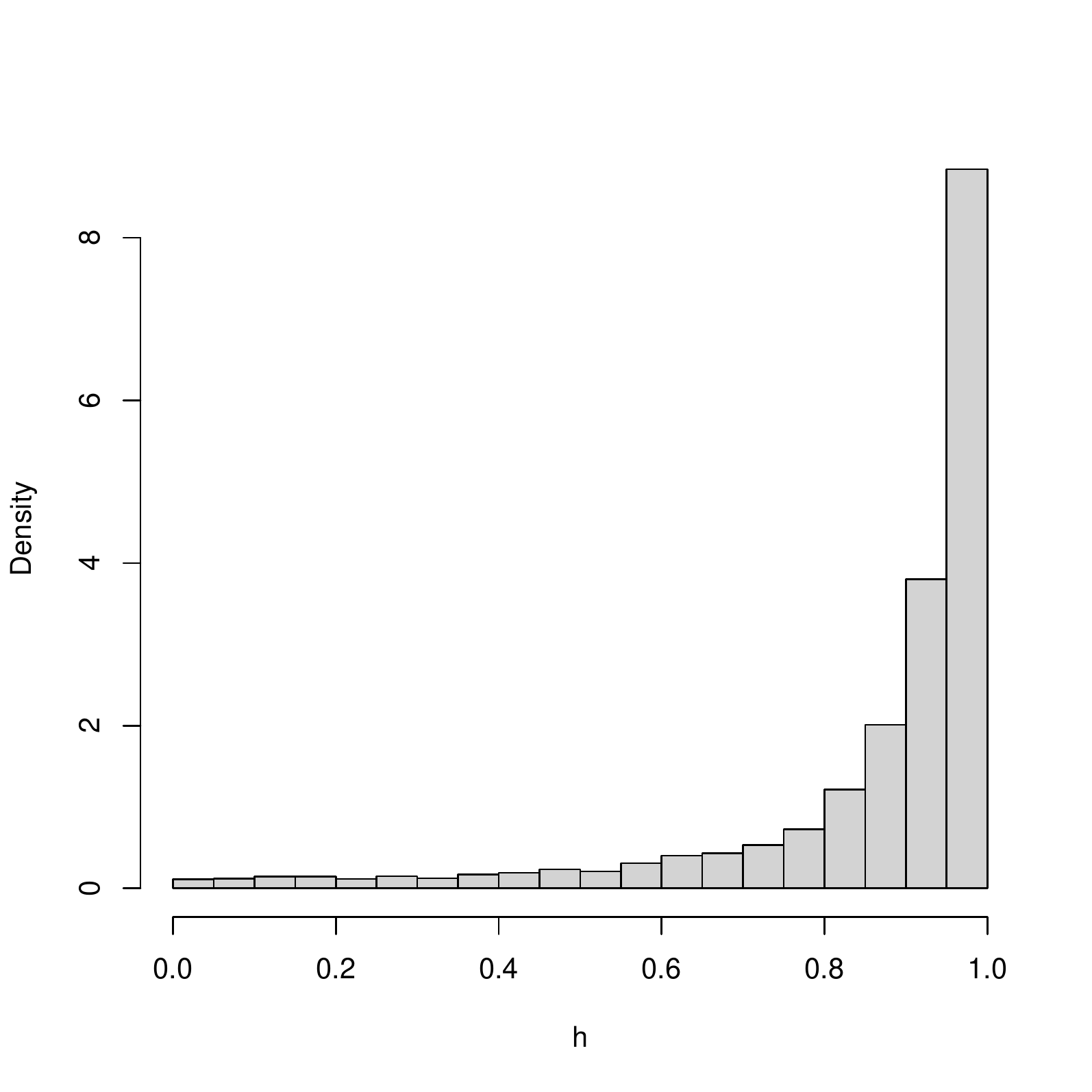} &
\includegraphics[width=0.4\textwidth]{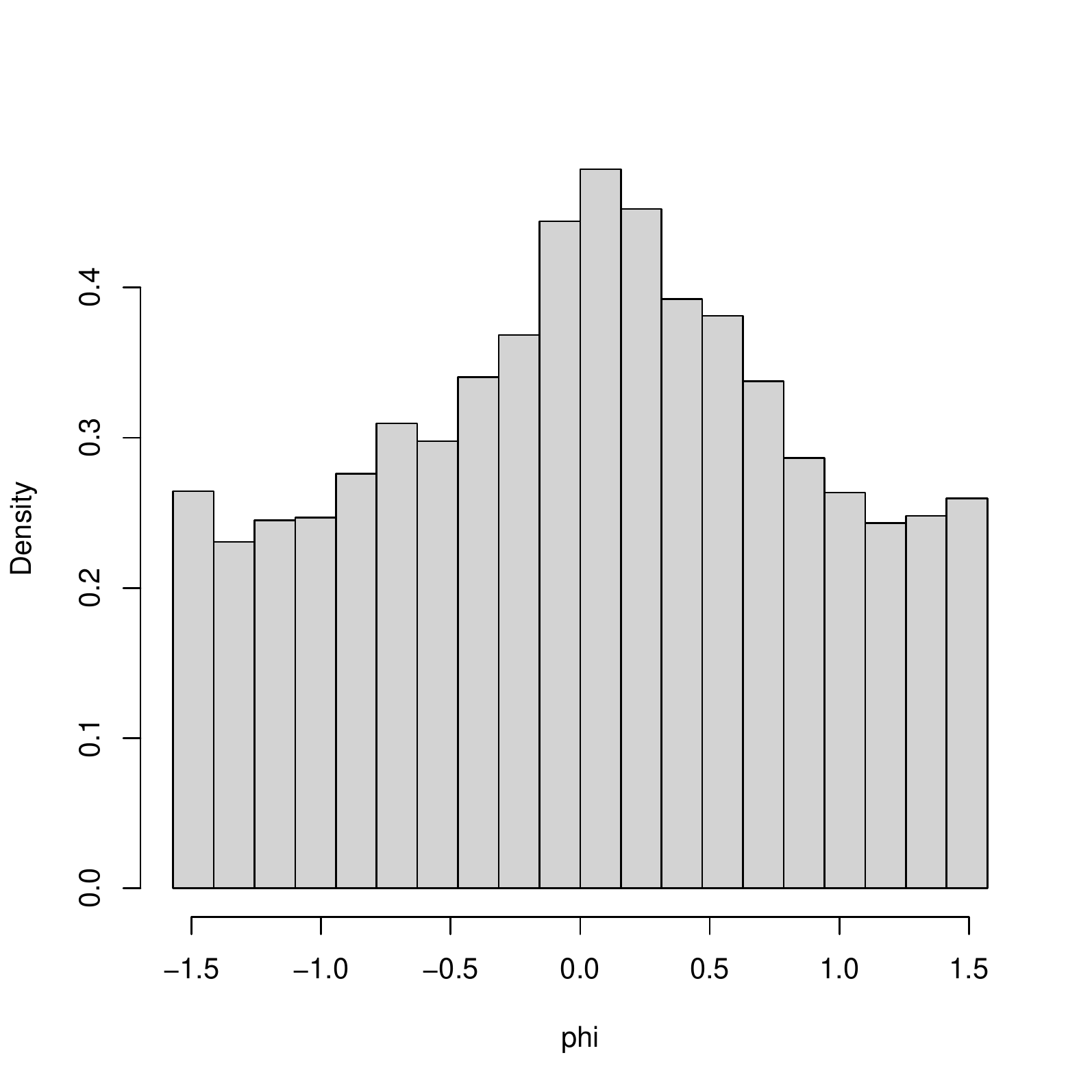} 

\end{tabular}
  \caption{Left: histogram of absolute values of $x$-coordinates of
    unit tangent vectors. Right: angle-components (in $(y,z)$-plane) of cylindrical coordinates of
  unit tangent vectors.  }\label{fig:fiber2}
\end{figure}

The histograms show that the distribution of angles in the
$(y,z)$ plane is roughly symmetric around 0  with  moderate deviation from the uniform distribution while the distribution of
heights assigns most mass at heights ($x$-coordinate of unit
tangent vectors) close to one. This means that the
fiber tangents are in general aligned with the $x$ direction as
desired for high strength of the concrete.

\subsection{The $K$-function for the concrete data}
We estimated the inhomogeneous $K_w$-function by 
\eqref{eq:Kestimatefiber}--\eqref{eq:montecarlo} using the intensity
estimate described in Section~\ref{sec:density_est} as weights. For
comparison, we also estimated the $K_w$-function assuming a constant
estimated intensity. 
The results are shown in Figure \ref{fig:Kdata}.  
\begin{figure}
	\centering
  \includegraphics[width=\textwidth]{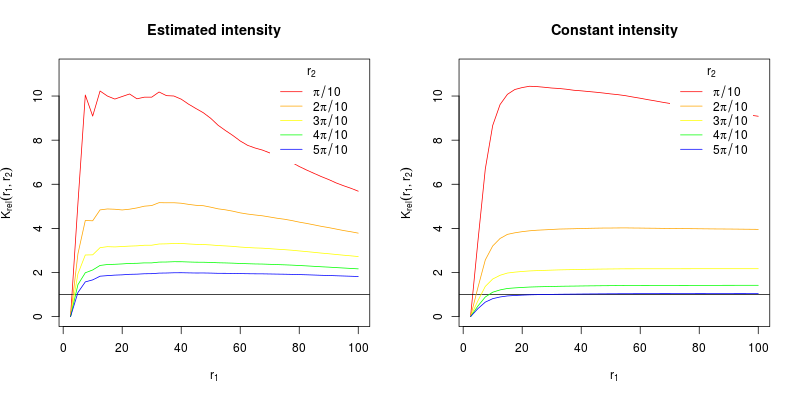}
	\caption{Relative $K$-function \eqref{eq:relK} for the concrete data set as a function
          of $r_1$ for five different values of $r_2$. Left: 
            with non-constant estimated intensity. Right: with constant estimated intensity.} \label{fig:Kdata}
\end{figure}
For the smallest angle $r_2=\pi/10$, the $K_w$-function is up to 10
times that of the null model. As $r_2$ increases, the $K_w$-function
comes closer to the null  model. This suggests that fiber directions
tend to be locally more similar than under the null model. If
we do not correct for inhomogeneity, we see the same effect. However,
in this case, we cannot tell whether it is due to first-order
inhomogeneity or the local alignment of fibers.  For the
$K_w$-function taking into account inhomogeneity, the deviation from the null model seems to decrease with $r_1$, suggesting that fiber alignment is strongest over small distances. This effect is not nearly as pronounced when constant intensity is used, possibly because the preferred fiber direction is present at all distances.

   To investigate whether deviations from the null model could be due
   to chance, we simulated 39 fiber patterns from the null model by
   generating fiber midpoints from a Poisson point process with
   spatial intensity $\rho_Y$ estimated from the original
   dataset. To estimate $\rho_Y$, we assumed a linear model as in
   \ref{sec:reflectioninvariant} and used
   \eqref{eq:reflection_density} to obtain $\rho_Y(u)$ as
   $\eta(s)^{-1}l^{-1}\rho(u,s;\beta)$ where for
   $\rho$ and $\eta$ we used the estimates from Section~\ref{sec:density_est} and for $l$ the average fiber length.  For each midpoint, a fiber was sampled with replacement from the fibers in the original dataset. To compute the $K_w$-function, we reestimated $\rho$ by the same procedure as used for the actual dataset. The pointwise minimum and maximum values provide pointwise 95\% envelopes for the relative $K$-function, see Figure  \ref{fig:envelopes}. The relative $K$-function estimated from data is clearly outside the envelope for all the displayed values of $r_2=\pi/10, 3\pi/10, \pi/2$, suggesting a deviation from the null model.

\begin{figure}
	\centering
	\includegraphics[width=\textwidth]{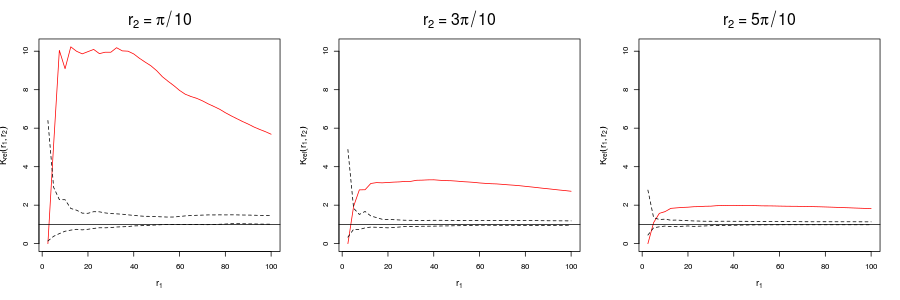}
	\caption{The relative $K$-function \eqref{eq:relK} (red curves) for three different values of $r_2$ with pointwise 95\% envelopes (dashed curves).}\label{fig:envelopes}
\end{figure}

\section{Discussion}

In this paper we developed a new $K_w$-function  for studying random
measures of the germ-grain type allowing to take geometrical features of grains into account. In particular, we considered the
case of fiber patterns with associated fiber tangents. The idea of
using geometrical information in second-order analyses of spatial patterns
is not new \citep{schwandtke:88} but to the best of our knowledge, this
idea has not before materialized in analyses of simulated and real data. In
contrast to \cite{schwandtke:88}, we moreover allow for an
inhomogeneous density of the random measure. As demonstrated by our
simulation studies and data example, taking into account inhomogeneity
(be it spatially or with respect to the distribution of the
geometrical features)
may have profound implications for the interpretation of estimated
$K_w$-functions. In contrast to \cite{gallego:ibanez:simo:16}, our
assumption of distinguishable components of the random measures allow
us to precisely characterize our $K_w$-function under the null model of
independent germ-grains. This is useful as a benchmark for interpreting an
estimated $K_w$-function.

Although we use a germ-grain/marked point process framework for our
mathematical derivations, we emphasize that only knowledge of the
grains/marks are required for our methodology. If the germs/points
were known one might consider alternatives such as the mark-weighted
$K$-function \citep{penttinen:stoyan:henttonen:92} that could be
generalized to the case with an inhomogeneous density for the
points. However, in addition to requiring observation of the points
this would also require full observation of the marks which can be
problematic due to edge effects.

Regarding estimation of first-order properties of inhomogeneous
  random measures, we constructed new estimating functions for fitting
  parametric models for the first order density $\rho$. This
  methodology is in fact applicable to general inhomogeneous random
  measures not necessarily of the germ-grain type considered in this paper.

A practical issue in relation to applying our $K_w$-function  is that
identifying the individual grains (e.g.\ fibers) from image data may not be
straightforward. For our data, this required a substantual
amount of manual work and there is certainly scope for research in improved image processing algorithms.
A theoretical topic for future research is to widen the catalogue of
specific models for which the reweighted second-order measure
$\af_w^{(2)}$ is known to be stationary.\\[\bsl]
\noindent {\bf Acknowledgments} Hans J.\ T.\ Stephensen was supported by QIM - Center for Quantification of Imaging Data from MAX IV.
The concrete fiber data and Figure~\ref{fig:concrete} was kindly provided by Claudia Redenbach, Technische Universit{\"a}t Kaiserslautern. The concrete fiber data is a result of joint work by Kasem Maryamh (sample), Technische Universit{\"a}t Kaiserslautern,  Franz Schreiber (imaging), Fraunhofer ITWM, and Konstantin Hauch (binarization), Technische Universit{\"a}t Kaiserslautern.
\bibliographystyle{rss}
\bibliography{fibres}

\begin{thebibliography}{20}
\expandafter\ifx\csname natexlab\endcsname\relax\def\natexlab#1{#1}\fi
\expandafter\ifx\csname url\endcsname\relax
  \def\url#1{\texttt{#1}}\fi
\expandafter\ifx\csname urlprefix\endcsname\relax\def\urlprefix{URL: }\fi

\bibitem[{Baddeley et~al.(2015)Baddeley, Rubak and
  Turner}]{baddeley:rubak:turner:15}
Baddeley, A., Rubak, E. and Turner, R. (2015) \textit{Spatial Point Patterns:
  Methodology and Applications with {R}}.
\newblock London: Chapman and Hall/CRC Press.

\bibitem[{Baddeley et~al.(2000)Baddeley, M{\o}ller and
  Waagepetersen}]{baddeley:moeller:waagepetersen:00}
Baddeley, A.~J., M{\o}ller, J. and Waagepetersen, R. (2000) Non- and
  semi-parametric estimation of interaction in inhomogeneous point patterns.
\newblock \textit{Statistica Neerlandica}, \textbf{54}, 329--350.

\bibitem[{Chiu et~al.(2013)Chiu, Stoyan, Kendall and Mecke}]{chiu:etal:13}
Chiu, S., Stoyan, D., Kendall, W. and Mecke, J. (2013) \textit{Stochastic
  Geometry and Its Applications}.
\newblock Wiley Series in Probability and Statistics. Wiley.

\bibitem[{Gallego et~al.(2016)Gallego, Ibáñez and
  Simó}]{gallego:ibanez:simo:16}
Gallego, M.~A., Ibáñez, M. V.~V. and Simó, A. (2016) Inhomogeneous
  {$K$}-function for germ–grain models.
\newblock \textit{Spatial Statistics}, \textbf{18}, 489 -- 504.

\bibitem[{Hansen et~al.(2021)Hansen, Waagepetersen, Svane, Sporring,
  Stephensen, Hasselholt and Sommer}]{hartung:etal:21}
Hansen, P. E.~H., Waagepetersen, R., Svane, A.~M., Sporring, J., Stephensen, H.
  J.~T., Hasselholt, S. and Sommer, S. (2021) Generalizations of {R}ipley's
  {$K$}-function with application to space curves.
\newblock In \textit{Proceedings of the GSI'21}.

\bibitem[{Jensen et~al.(1990)Jensen, Kiêu and
  Gundersen}]{jensen:kieu:gundersen:90}
Jensen, E., Kiêu, K. and Gundersen, H. (1990) On the stereological estimation
  of reduced moment measures.
\newblock \textit{Ann Inst Stat Math}, \textbf{42}, 445–461.

\bibitem[{{V}an Lieshout(2018)}]{lieshout:18}
{V}an Lieshout, M. N.~M. (2018) Nonparametric indices of dependence between
  components for inhomogeneous multivariate random measures and marked sets.
\newblock \textit{Scandinavian Journal of Statistics}, \textbf{45}, 985--1015.

\bibitem[{Maryamh et~al.(2021)Maryamh, Hauch, Redenbach and
  Schnell}]{maryamh:etal:21}
Maryamh, K., Hauch, K., Redenbach, C. and Schnell, J. (2021) Influence of
  production parameters on the fiber geometry and the mechanical behavior of
  ultra high performance fiber-reinforced concrete.
\newblock \textit{Structural Concrete}, \textbf{22}, 361--375.

\bibitem[{M{\o}ller and Waagepetersen(2017)}]{moeller:waagepetersen:17}
M{\o}ller, J. and Waagepetersen, R. (2017) Some recent developments in
  statistics for spatial point processes.
\newblock \textit{Annual Review of Statistics and its Applications},
  \textbf{4}, 317--342.

\bibitem[{M{\o}ller and Waagepetersen(2004)}]{MoellerWaagepetersenl2004}
M{\o}ller, J. and Waagepetersen, R.~P. (2004) \textit{Statistical Inference and
  Simulation for Spatial Point Processes}.
\newblock {Boca Raton}: {CRC Press}.

\bibitem[{M{\o}ller and Waagepetersen(2007)}]{moeller:waagepetersen:07}
--- (2007) Modern statistics for spatial point processes.
\newblock \textit{Scandinavian Journal of Statistics}, \textbf{34}, 643--684.

\bibitem[{Penttinen et~al.(1992)Penttinen, Stoyan and
  Henttonen}]{penttinen:stoyan:henttonen:92}
Penttinen, A., Stoyan, D. and Henttonen, H.~M. (1992) Marked point processes in
  forest statistics.
\newblock \textit{Forest Science}, \textbf{38}, 806--824.

\bibitem[{Ripley(1976)}]{ripley:76}
Ripley, B.~D. (1976) The second-order analysis of stationary point processes.
\newblock \textit{Journal of Applied Probability}, \textbf{13}, 255--266.

\bibitem[{Schneider et~al.(2012)Schneider, Rasband and
  Eliceiri}]{schneider:rasband:eliceiri:12}
Schneider, C.~A., Rasband, W.~S. and Eliceiri, K.~W. (2012) {NIH} image to
  {ImageJ}: 25 years of image analysis.
\newblock \textit{Nature Methods}, \textbf{9}, 671–675.

\bibitem[{Schwandtke(1988)}]{schwandtke:88}
Schwandtke, A. (1988) Second-order quantities for stationary weighted fibre
  processes.
\newblock \textit{Mathematische Nachrichten}, \textbf{139}, 321--334.

\bibitem[{Shaw et~al.(2020)Shaw, M{\o}ller and
  Waagepetersen}]{shaw:moeller:waagepetersen:20}
Shaw, T., M{\o}ller, J. and Waagepetersen, R. (2020) Globally
  intensity-reweighted estimators for {$K$}- and pair correlation functions.
\newblock \textit{Australian and New Zealand Journal of Statistics},
  \textbf{63}, 93--118.

\bibitem[{Stoyan and Ohser(1982)}]{stoyan:ohser:82}
Stoyan, D. and Ohser, J. (1982) Correlations between planar random structures
  with an ecological application.
\newblock \textit{Biometrical Journal}, \textbf{24}, 631--647.

\bibitem[{Virtanen et~al.(2020)Virtanen, Gommers, Oliphant, Haberland, Reddy,
  Cournapeau, Burovski, Peterson, Weckesser, Bright, {van der Walt}, Brett,
  Wilson, Millman, Mayorov, Nelson, Jones, Kern, Larson, Carey, Polat, Feng,
  Moore, {VanderPlas}, Laxalde, Perktold, Cimrman, Henriksen, Quintero, Harris,
  Archibald, Ribeiro, Pedregosa, {van Mulbregt} and {SciPy 1.0
  Contributors}}]{2020SciPy-NMeth}
Virtanen, P., Gommers, R., Oliphant, T.~E., Haberland, M., Reddy, T.,
  Cournapeau, D., Burovski, E., Peterson, P., Weckesser, W., Bright, J., {van
  der Walt}, S.~J., Brett, M., Wilson, J., Millman, K.~J., Mayorov, N., Nelson,
  A. R.~J., Jones, E., Kern, R., Larson, E., Carey, C.~J., Polat, {\.I}., Feng,
  Y., Moore, E.~W., {VanderPlas}, J., Laxalde, D., Perktold, J., Cimrman, R.,
  Henriksen, I., Quintero, E.~A., Harris, C.~R., Archibald, A.~M., Ribeiro,
  A.~H., Pedregosa, F., {van Mulbregt}, P. and {SciPy 1.0 Contributors} (2020)
  {{SciPy} 1.0: Fundamental Algorithms for Scientific Computing in Python}.
\newblock \textit{Nature Methods}, \textbf{17}, 261--272.

\bibitem[{Waagepetersen(2007)}]{waagepetersen:07}
Waagepetersen, R. (2007) An estimating function approach to inference for
  inhomogeneous {N}eyman-{S}cott processes.
\newblock \textit{Biometrics}, \textbf{63}, 252--258.

\bibitem[{Waagepetersen and Guan(2009)}]{waagepetersen:guan:09}
Waagepetersen, R. and Guan, Y. (2009) Two-step estimation for inhomogeneous
  spatial point processes.
\newblock \textit{Journal of the Royal Statistical Society, Series B},
  \textbf{71}, 685--702.

\end{thebibliography}

\appendix

\section{A simple fiber model}\label{sec:simplefibermodel}

A  random line segment $\Gamma $ generates the random measure 
\[ m(A \times S) = \int_{\Gamma} \ind [x \in A, \tau(x) \in S] \lambda(\dd
  x), A \subseteq \R^d, S \subseteq \Sg^{d-1}. \]
If $\Gm$ is centered at the origin, of random length
$L$, and of randomly  distributed orientation $\Theta \in
\Sg^{d-1}$, then 
\[ m(A \times S) = \int_{-L/2}^{L/2}  \ind [ r \Theta \in A, \Theta \in S] \dd
  r.\]

Assuming moreover that $L$ and $\Ta$ are independent and
  $\Theta$ has a density $\eta$ with respect to $\nu$, we obtain
\[ \EE m(A \times S)=  \int_{\Sg^{d-1}}
  \E \int_{-L/2}^{L/2} \eta(\ta) \ind [ r \theta \in A, \theta \in S] \dd r  \nu (\dd \ta).
  \]
  
Consider a fiber process where, given the underlying point process
$Y$, the marks are line segments identically distributed as $\Gamma$ translated to have midpoints in $Y$.
Then,
\begin{align*}
\mu(A \times S) = & \int_{\R^d} \rho_Y(u) \EE[ m((A-u) \times S)] \dd u
  \\
= & \int_{\Sg^{d-1}} \ind [ \theta \in S]
    \int_{\R^d} \rho_Y(u)
\E  \int_{-L/2}^{L/2}  \eta(\ta) \ind [ r \theta +u \in A]
     \dd r \dd u \nu( \dd \ta) \\
= &  \int_{\Sg^{d-1}} \int_{\R^d} \ind [ z \in A, \theta \in S]
 \E \int_{-L/2}^{L/2}  \eta(\ta) \rho_Y(z-r \theta) \dd r \dd z \nu( \dd \ta).
\end{align*}
Hence $\mu$ has density
\[ \rho(z,s)= \eta(s)\E \int_{-L/2}^{L/2}   \rho_Y(z-r
  s) \dd r .\]
  
Suppose $\rho_Y$ is given in terms of a linear model $\rho_Y(u)= \bt_0 + \bt^\T u$.
Then 
\[ \rho(z,s) = \eta(s)\E \int_{-L/2}^{L/2}  [\bt_0 + \bt^\T (z-r s)] \dd
  r = \eta(s) l (\bt_0+ \bt^\T z),  \]
  where $l=\E L$.

\section{An alternative choice of weights }\label{sec:w2}
For the weights $w(u,x,s)=\xi_u(x,s)\rho_Y(u)$ suggested in Remark
\ref{rem:w2} assume that $\xi_u(\cdot ,s)$ has
a bounded support $u+C$ where $C$  does not depend on $u$ or
$s$. Then $\mu_{w}(A \times S)= |A| \nu(S) |C|$.
Assuming that $\xi_{u,u'}$ exists, define
\[ g_{u,u'}(x,s,y,t) =
  \frac{\xi_{u,u'}(x,s,y,t)}{\xi_u(x,s)\xi_{u'}(y,t)}\ind[x-u,y-u'\in C].\]
By \eqref{eq:campbell_alpha} and \eqref{eq:xiuu'},
\begin{align*}& \EE  \sum_{(u,m),(u',m') \in X}^{\neq} \int_{(\R^d
                \times \Sg)^2}f(x,s,y,t) m(\dd x \times \dd
                s)m'(\dd y \times \dd t) \nonumber \\=& 
  \int_{(\R^d \times \Sg)^2}f(x,s,y,t) \int_{\R^d \times \R^d} \rho_Y(u)\rho_Y(u') g_Y(u,u') 
                                            \xi_{ u ,
                                            u'}((x,s),(y,t)) \nonumber \\ 
                        &\dd u \dd u' \dd x \dd
                                            y \nu(\dd s) \nu(\dd t).
\end{align*}
We thus have the following expression for $K_w$:
\begin{align*}
  &K_{w}(r_1,r_2;W) \\
   =
    &\frac{1}{|W|}\int_{\R^d \times \R^d }\int_{(C+u)\times (C+u')\times \Sg^2 }
  \ind [x \in W] \ind[ \|x  -y\|\leq r_1,
  d_{\Sg} (s,t) \le r_2] \\ & \cdot  g_{u,u'}(x,s,y,t)g_Y(u,u')
  \dd x \dd y
  \nu(\dd s) \nu(
  \dd t) \dd u \dd u'.\nonumber 
  \end{align*}
  
Stationarity of $\alpha_w^{(2)}$ is obtained if for all $h\in \R^d$
\begin{equation}\label{eq:bothd} g_Y(u,u')=g_Y(u+h,u'+h)  \,\, \text{and} \,\,
  g_{u,u'}(x,s,y,t)=g_{u+h,u'+h}(x+h,s,y+h,t). 
   \end{equation}
In this case, 
\begin{align*}
  K_{w}(r_1,r_2;W)   = &\int_{  \R^d\times (C\times \Sg )^2 }
  \ind [\|x-y-l\| \le r_1,d_{\Sg}(s,t)\le r_2 ]\\
  &\cdot g_{0,l}(x,s,y+l,t)g_Y(0,l)                                                      \dd x \dd y
                                                      \nu(\dd s) \nu(
                                                      \dd t) \dd l.\nonumber 
   \end{align*}

The first condition in \eqref{eq:bothd} requires that $Y$ is
second-order intensity re\-weigh\-ted stationary which is
a fairly weak requirement, see \cite{baddeley:moeller:waagepetersen:00}. 
The second requirement is for instance satisfied if
$\xi_u(x,s)=\xi(x-u,s)$ for a common $\xi$ not depending on $u$ and
if, similarly, $\E_{u,u'} m(\cdot)m'(\cdot)$ is translation
  invariant, $\E_{u,u'} m(\cdot)m'(\cdot)=\E_{u+h,u'+h} m(\cdot+h)m'(\cdot+h)$, $h \in \R^d$.
This is for instance the case for the model with dependent fibers used for simulation in Section \ref{sec:simulations}. 

Under the null model, $g_Y=1$ and $\xi_{u,u'}=\xi_u \xi_{u'}$.  Then $g_{0,l}(x,s,y,t)g_Y(0,l)= \ind [ x , y \in C]$ and 
\[ K_{w}(r_1,r_2)= K_{0}(r_1,r_2) |C|^2, 
\]
where $K_0$ is as in \eqref{eq:K_null}.
In practice $|C|$ is not known. Given an estimate $\hat K_{w}$ we might
consider the ratio $\hat K_{w}(r_1,r_2)/K_0(r_1,r_2)$ which should be close to
a constant over $r_1$ and $r_2$ under the null model.

For more general models, if $g_{u,u'}(x,s,y,t)\leq (\geq) 1$ and $g_Y(u,u')\leq (\geq) 1$ for all $u,u'\in \R^d$ and all $(x,s,y,t)\in (C\times\R^d)^2$, then $K_{w}(r_1,r_2)\leq (\geq ) K_{0}(r_1,r_2)|C|^2$ for all $r_1,r_2$. 


\section{A model with reflection invariant mean mark density}\label{sec:reflectioninvariant}

Suppose that $\xi_u(x,s)=\xi(x-u,s)$ where $\xi$ satisfies $\xi(x,s)=\xi(-x,s)$. This is
for example satisfied is $\xi(\cdot,s)$ is a density with elliptical contours centered
at the origin.
With 
\[ \rho_Y(u) = \bt_0 + \bt^\T u \]
and defining 
\[ \tilde \xi (s)= \int_{\R^{d}}  \xi(x,s) \dd x, \]
we obtain
\begin{align*} \rho(z,s) = & \frac{1}{2}\int_{\R^d} (  \xi(x,s) +  \xi
  (-x,s))(\bt_0 + \bt^\T [z-x] )  \dd x \\ = & (\bt_0 + \bt^\T z)
  \int_{\R^{d}}  \xi(x,s) \dd x + \frac{1}{2} \int_{\R^d}
   \xi(x,s)  \bt^\T [-x + x] \dd x \\ = & (\bt_0 + \bt^\T z) \tilde \xi(s). 
\end{align*}
Defining
\[ l= \int_{\Sg^{d-1}} \tilde \xi(s) \nu( \dd s) \]
and the probability density
\[ \eta(s)= \tilde \xi(s) / l, \]
we obtain 
\begin{equation}\label{eq:reflection_density}
 \rho(z,s)= \eta(s)l \rho_Y(z) .
 \end{equation}

Suppose now that $\rho_Y(z)$ is non-linear in $z$, that $\xi(\cdot,\cdot)$ has
bounded support $C$ with respect to the
first argument, and that $\rho_Y(\cdot)$ is well approximated by its
tangent plane on $C+z$. The
convolution
\[ \rho(z,s)= \int_{\R^d} \xi(z-u,s)\rho_Y(u) \dd u \]
only involves $\rho_Y(u)$ for $u \in C+z$. Hence we obtain
\[ \rho(z,s) \approx \eta(s)l \rho_Y(z) \]
following the arguments leading to \eqref{eq:reflection_density}.

\section{A model with dependent fibers}\label{sec:simplefibermodelinteract}

We now describe the model used for simulation of dependent fibers in Section \ref{sec:simulations}. We again choose $Y$ as a Poisson process and generate independent stationary zero-mean
Gaussian fields $X_1,\ldots,X_d$ such that $\Var X_i(u)=\Var X_j(u)$, $u
\in \R^d$, $i,j=1,\ldots,d$. For a location $u$, we further let
$X(u)=(X_1(u),\ldots,X_d(u))$, $L_u=\|X(u)\|$ and $T_u=X(u)/L_u$. Then
$T_u$ has uniform density $\eta$ on $\Sg^{d-1}$ and is independent
 of
$L_u$. Each point $u\in Y$ is marked by the line segment $\Gm_u$  from
$-X(u)$ to $X(u)$ translated by $u$. The fiber measure 
$m$ associated to $\Gamma_u$ is 
\[ m(A \times S) = \int_{\Gm_u}  \ind[ z \in A,T_u \in S ] \lambda (\dd
  z) = L_u  \int_{-1}^{1} \ind [ u+ t T_u L_u \in A, T_u \in S] \dd t. \]
Since the distribution of $\Gamma_u $ is independent of $u$, we get from \ref{sec:simplefibermodel} that
 $\mu$ has density 
\[ \rho(z,s)=\eta(s) \int_{-L_u}^{L_u} \EE[
\rho_Y(z-r s ) ] \dd r\]
and if $\rho_Y$ is linear 
$ \rho_Y(u) = \bt_0 +\beta^\T u ,$
then
\begin{align*}
\rho(z,s)= \eta(s) l (\bt_0 +\bt^\T z)
\end{align*}
with $l=2  \EE L_u$.

\end{document}